\documentclass[a4paper,10pt]{article}

\usepackage{algorithm, algpseudocode,color, graphicx, lineno}
\newtheorem{theorem}{Theorem}[section]
\newtheorem{Lemma}[theorem]{Lemma}

\newenvironment{proof}[1][Proof]{\begin{trivlist}
		\item[\hskip \labelsep {\bfseries #1}]}{\end{trivlist}}

\begin{document}

\markboth{Chitturi and Pai}
{ Minimum-Link Rectilinear Covering Tour is NP-hard in $R^{4}$}

\title{Minimum-Link Rectilinear Covering Tour is NP-hard in $R^{4}$}
\author{$^{1,3}$ \small Bhadrachalam Chitturi and $^2$Jayakumar Pai\\\small$^1$Department of Computer Science and Engineering,\\\small Amrita Vishwa Vidyapeetham, Amritapuri, India\\
\small $^2$Department of Computer Science and Applications,\\ \small Amrita Vishwa Vidyapeetham, Amritapuri, India\\\small
$^3$Department of Computer Science,\\\small University of Texas at Dallas, Richardson, TX, USA
}

\date{}
\maketitle
%\linenumbers

\begin{abstract}
Given a set $P$ of $n$ points in $R^{d}$, a tour is a closed simple path that covers all the given points, i.e. a Hamiltonian cycle. % In $P$ if no three points are collinear then the points are said to be in general position. 
A \textit{link} is a line segment connecting two points and a rectilinear link is parallel to one of the axes. 
The problems of defining a path and a tour with minimum number of links, also known as Minimum-Link Covering Path and Minimum-Link Covering Tour respectively are proven to be NP-hard in $R^2$. The corresponding rectilinear versions are also NP-hard in $R^2$. 

 A set of points is said to be in \textit{general position} for rectilinear versions of the problems if no two points share any coordinate. We call a set of points in $R^{d}$ to be in \textit{relaxed general position} if no three points share any coordinate and any two points can share at most one coordinate. That is, if the points are either in general position or in relaxed general position then an axis parallel line can contain at most one point. 
If points are in relaxed general position then these problems are NP-hard in $R^{10}$.  We prove that these two problems are in fact NP-hard in $R^{4}$. 
If points in $R^{d},~d>1$ are in general position then the time complexities of these problems, both basic and rectilinear versions, are unknown.\\
\emph{Keywords}:{Covering Path, Covering Tour, Rectilinear Covering, Computational Complexity, Vector Space, Computational Geometry.}
\end{abstract}
%A set of points is said to be in \textit{general position} if no three points are collinear.

%\ccode{1991 Mathematics Subject Classification: 68Q25, 68R10, 52C45, 68U05}
%68U05   	Computer graphics; computational geometry
%52C45   	Combinatorial complexity of geometric structures
% 05C30,11Y16, 68R05, 68R10,

\section{Introduction}
The problem of covering a finite set of points with line segments is a fundamental problem in computational geometry with applications in fields like VLSI design, manufacturing etc.. This article addresses the problem of covering points with tours comprising of rectilinear line segments, which are axis parallel line segments. Let $ P =\{p_{1}, p_{2}, p_{3}... p_{k}\}$ be the given set of input points in $R^{d}$. A $link$ is a line segment that connects two points possibly from $P$. A $chain$ is a simple path of links that passes through all the points of $P$ whereas a $closed$-$chain$ is a simple circuit that passes through all points of $P$.
An intermediate point on a path is a turn or a bend \cite{stein2001approximation, jiang2015covering}. Two adjacent links are connected at a common point which is a $turn$.  
A $spanning\; path$ is a chain where the links connect points from P only. That is, a spanning path can take turns only at given input points. 
A $covering\;path$ is a polygonal chain which can turn at chosen points that are not in $P$.  Dumitrescu et al.\cite{dumitrescu2014covering} have shown that the number of links required for a spanning path is $n-1$ and at least $\left \lceil{\frac{n}{2}}\right \rceil$ links are required for a covering path.  The bounds for the number of links in a \textit{ minimum link spanning path problem} are studied in \cite{bereg2009traversing} when the paths are axes aligned.
The general covering problems include covering a given set of points in $R^{d}$ with line segments or with a chain or with a closed chain. These problems are respectively called \textit{Line Cover Problem, Covering Path Problem} and \textit{Covering Tour Problem}.  If the links are restricted to be axis parallel then we obtain the rectilinear versions of these problems. The general and rectilinear versions of these covering problems are posed as optimization problems where the number of turns are to be minimized.  A summary of the complexity results of these problems is given in Table 1.

 Covering problems with constraints on the points also have been investigated. A set of points is said to be in \textit{general position} if no three points are collinear  \cite{dumitrescu2014covering}. In this article we address rectilinear versions of these problems. For rectilinear versions a set of points is said to be in \textit{general position} if no two points share any coordinate \cite{jiang2015covering}. 
	The complexities of \textit{Covering Path Problem} and \textit{Covering Tour Problem} when the points are in general position are not known \cite{jiang2015covering}. 
	We call a set of points in $R^{d}$ to be in \textit{relaxed general position} if no three points share any coordinate and any two points can share at most one coordinate.	
	If a set of points in rectilinear versions of the problems are in either  general position or relaxed general position then they have at most one point on an axis parallel line for $d>2$. 	
When points are in relaxed general position Jiang \cite{jiang2015covering} proved that \textit{Minimum-Link Rectilinear Covering Tour} and \textit{Minimum-Link Rectilinear Covering Path} are NP-hard in $R^{10}$. In this article we show that these problems are NP-Hard in $R^{4}$.
%The points we construct in $R^{4}$ are in relaxed general position.	

\begin{table}[ht]
	\caption{A list of standard covering problems and known complexity results} % title of Table
	\centering % used for centering table
	\begin{tabular}{|c |c|c| } % centered columns (4 columns)
		\hline\hline 
		\textbf{SlNo} &\textbf{ Problem} &\textbf{ Result} \\ [0.25ex] % inserts table
		%heading
		\hline % inserts single horizontal line
		1 & Line Cover &	NPH in $R^{2}$\cite{megiddo1982complexity} \\ % inserting body of the table
		\hline
		2 &Rectilinear Line cover &	P in $R^{2}$, NPH in $R^{3}$ \cite{hassin1991approximation} \\
		\hline
		3& Minimum-Link Covering Path &	NPH in  $R^{2}$\cite{kranakis1994link} \\
		\hline
		4 & Minimum-Link Covering Tour &	NPH in $R^{2}$\cite{arkin2003minimum}  \\
		\hline
		5 &Minimum-Link Rectilinear Covering Path  & NPH in $R^{10}$\cite{jiang2015covering}\\
		\hline
		6&Minimum-Link Rectilinear Covering Tour &NPH in $R^{10}$\cite{jiang2015covering}\\
		\hline
		
	\end{tabular}
	\label{table:nonlin} % is used to refer this table in the text
\end{table}

The remainder of the article is organized as follows. Section 2 discusess the background. In Section 3 construction of points in $R^4$ from the corresponding points in $R^2$ grid is shown and the complexity proof of Min-Link Rectilinear Tour problem is given. Section 4 proves the complexity of Min-Link Rectilinear Path problem. Section 5 states the conclusions.

\section{Previous methods} 
The angle between adjacent links can be either $90^\circ$, $180^\circ$ or $270^\circ$ at any turn for rectilinear versions of the problems \cite{bereg2009traversing,jiang2015covering}.
Wang \textit{et al}. \cite{wang2014minimum} proved that the problems Minimum-Link Rectilinear Covering Tour and Minimum-Link Rectilinear Covering Path are both NP-Hard in $R^{2}$. Wang's proof relies on a reduction from cardinality constrained bipartite graph problem and in his method mutiple points are covered by a rectilinear line segment. 
Jiang's proof relies on a reduction from Hamiltonian cycles in grid graphs. In his construction no axis-parallel line in $R^{10}$ covers more than one constructed point. In addition to this, the constraints applied to the problems considered by Wang and Jiang are different. In the problem considered by the former there is no constraints on the points but for latter the points are in relaxed general position. In Jiangs method two points in $R^{10}$ shares a coordinate if the corresponding point in the grid are adjacent and hence the constructed points in $R^{10}$ are in relaxed general position.\\

Jiang used ten dimensions to represent the adjacencies and turns of points in the grid. He used four axes  $(a,\, b,\, c$ and $ d)$ in $R^{10}$  to represent the four adjacencies a point in a grid can have. He, then showed that a grid point can participate in exactly one of the six types of turns (shown in figure 1). He used six more axes  $(r,\, s,\, t,\, u,\, v$ and $ w)$ to represent the turn in which a point participates in such a way that if a point $p_{i}$ participate in a turn of type say 3, then the line segment which covers the corresponding point $q_{i} \in R^{10}$ is drawn parallel to one of the six axes say $t$, which represents turn of type 3. 
\begin{figure}[h!]
	\centering
	\includegraphics[width=0.7 \linewidth]{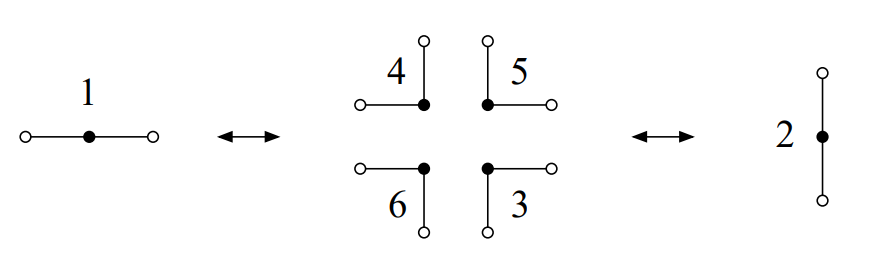}
	\caption{Six types of turns}
\end{figure}

In our method we use the first four dimensions $(a,\, b,\, c$ and $ d)$ used by Jiang in the same way to represent the adjacencies of grid points but the way we build the path and tour on the constructed points in $R^{4}$ differ from the same done by Jiang on constructed points in $R^{10}$. The subsequent sections discuss the rectlinear versions of the problems in relaxed general position.

\section{Minimum-Link Rectilinear Covering Tour in  $R^{4}$} 
$G (P, E)$ is a grid graph or simply a grid in the sequel, consists of a finite set $P$ of points in $R^{2}$ and every pair of adjacent points are connected by an edge. Two points $p_{i} = (x_{i}, y_{i})$  and $ p_{k} = (x_{k}, y_{k})$ are considered adjacent in the grid if $x_{i} = x_{k}$ and $y_{i} = y_{k} \pm 1$ or $y_{i} = y_{k}$ and  $ x_{i} = x_{k}\pm 1$. Let $|P| = n $ and $|E| = m$, we label each vertex of $P$ with a unique value from $\{1, 2...n\}$ and each edge of $ E$ with a unique value from $\{1, 2...m\}$. Let $p_{i} = (x_{i}, y_{i}) \in P$ be a grid point and let $q_{i} \in Q$ be the corresponding point in $R^{4}$. 
Let $(a, b, c, d)$ be the four axes of $R^{4}$. We specify a construction that maps a given $p_{i}$ to a specific $q_{i} = (a_{i}, b_{i}, c_{i}, d_{i})$. Our construction method is similar to that of M Jiang in \cite{jiang2015covering}. We call each $q_{i}$ as a \textit{constructed point}.

We determine the coordinates of $q_{i}$ based on the adjacencies of $p_{i}$, as follows. Let $p_{k}$ be a point adjacent to $p_{i}$ and let $j$ be the label of the edge $(p_{i}, p_{k})$. For every $q_{i}$ all its coordinates are initially assigned a value of $i$. Then we update the four coordinates of $q_{i}$ based on which one of the following cases applies.  
\\If $(y_{k} = y_{i})$\\
\hspace*{2em} If $x_{i}$ is odd and $x_{k} = x_{i}-1$, or if $x_{i}$ is even and $x_{k} = x_{i} + 1$ then set $a_{i} = n + j$\\ 
\hspace*{2em}	If $x_{i}$ is odd and $x_{k} = x_{i}+1$, or if $x_{i}$ is even and $x_{k} = x_{i}-1$, then set $b_{i} = n + j$\\
If $(x_{k} = x_{i})$\\
\hspace*{2em}If $y_{i}$ is odd and $y_{k} = y_{i}-1$, or if $y_{i}$ is even and $y_{k} = y_{i} + 1$ then set $c_{i} = n + j $\\
\hspace*{2em}If $y_{i}$ is odd and $y_{k} = y_{i} +1$, or if $y_{i}$ is even and $y_{k} = y_{i}-1$, then set $d_{i} = n + j$
\- An example of the construction is shown in fig 1. 
\begin{figure}[h!]
	\centering
	\includegraphics[width=\linewidth]{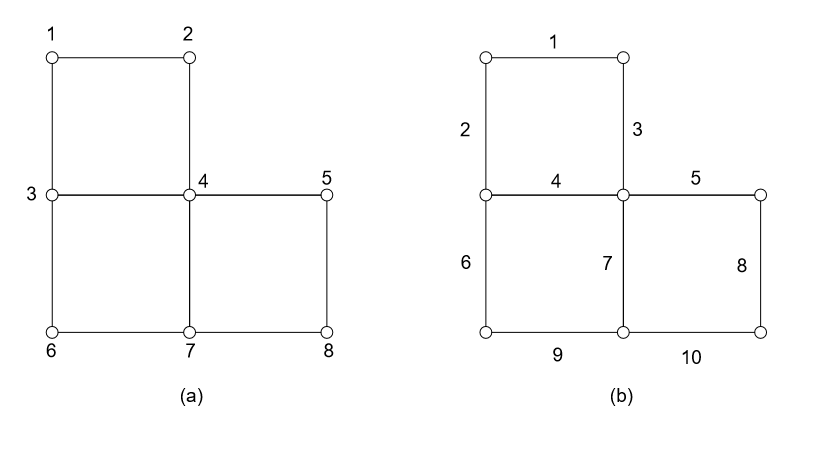}
	\includegraphics[width=\linewidth]{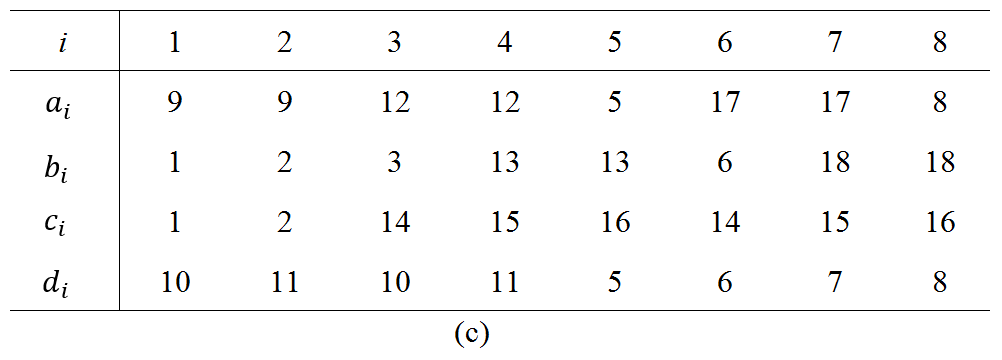}
	\caption{(a) Shows a sample grid in $R^{2}$ with labelled nodes and figure  (b) shows the labelling of edges. Figure (c) shows the coordinates of the constructed points in $R^{4}$}
\end{figure}
\begin{Lemma}
	\label{l1}
	The constructed points in $R^4$ are in relaxed general position.
\end{Lemma}
\begin{proof}
	Let $p_{i}$ and  $p_{k}$ be two points in $R^{2}$ and let $q_{i}$ and $q_{k}$ be the corresponding constructed points in $R^{4}$. (i) If $p_{i}$ and $p_{k}$ are adjacent then we show that exactly one coordinate of $q_{i}$ and $q_{k}$ has identical value. (ii) Otherwise, we show that $q_{i}$ and $q_{k}$ differ in all coordinates. (iii) If $p_{a}$,  $p_{b}$ and  $p_{c}$  are three consecutive points in $R^{2}$ then the coordinate shared by the constructed points $q_{a}$, $q_{b}$ differs from the coordinate shared by $q_{b}$,  $q_{c}$.
	
	(i) Consider two adjacent points $p_{i}$ and $p_{k}$ in $R^{2}$. Before update procedure $q_{i} = (i, i, i, i)$ and $q_{k} = (k, k, k, k)$. \\
	Case A: If $p_{i}$ and $p_{k}$ share the same $x$ value, then after update procedure, $q_{i} = (i, i, n+j, i)$ and   $q_{k} = (k, k, n+j, k)$ or $q_{i} = (i, i, i, n+j)$ and $q_{k} = (k, k, k, n+j)$.\\
	Case B: If $p_{i}$ and $p_{k}$ share the same $y$ value, then after update procedure, $q_{i} = (n+j, i, i, i)$ and  $q_{k} = (n+j, k, k, k)$ or $q_{i} = (i, n+j, i, i)$ and $q_{k} = (k, n+j, k, k)$. \\
	In both cases the points $q_{i}$ and $q_{k}$ share exactly one of the four axes.
	
	(ii)Assume that for some integers $u$ and $i (u\neq i)$ there exist two non-adjacent points $p_{u}$ and $p_{i}$ such that the corresponding constructed points $q_{u}$ and $q_{i}$ in $R^{4}$ share at least one coordinate. Without loss of generality, assume that the points $q_{u}$ and $q_{i}$ share the $a-axis$, i.e. $a_{u} = a_{i}$. 
	The following cases that update the coordinates prove that this happens only if the points are adjacent (contradiction) or $u = i$ (contradiction). During the construction the update procedure may update the axis a as follows.\\
	Case 1: Both $a_{u}$ and $a_{i}$ are updated. This happens only if the points are adjacent, which contradicts the assumption that they are non-adjacent.\\
	Case 2: $a_{u}$ is updated and $a_{i}$ is not updated. The coordinate $a_{u}$ is updated implies that $a_{u} \gets n+j$ where $j$ is label of the edge $(q_{u}, q_{i})$.  After updation $a_{u} = a_{i}$ implies $n+j=i$. As $n$ is the total number of points in the grid this case is not possible. \\
	Case 3: $a_{u}$ is not updated and $a_{i}$ is updated. Similar to Case 2.\\
	Case 4: Neither $a_{u}$ nor $a_{i}$ is updated Initially $a_{u} = u$ and $a_{i} = i$. So $a_{u} = a_{i}$ implies $u = i$. This contradicts the fact that $u \neq i$.
	
	(iii) Consider three consecutive points $q_a, q_b, q_c$ in $R^{4}$. The corresponding  adjacent pairs of points in the grid are ($p_a,\, p_b$) and ($p_b,\, p_c$). Let the coordinates of the three points $p_a,\, p_b,\,p_c$ be $(x_a,y_a), (x_b,y_b)$ and $(x_c,y_c)$ respectively. We prove the claim only for one axis (say, $a-axis$) since the proff for all other axes are similar. Assume that due to the adjacency of the points  $p_a$ and $p_b$ $a-axis$ is updated. Based on the construction this happens in following two cases only.\\  
	Case 1:. $y_a = y_b,\; x_a$ is odd and $x_b=x_a-1$ or\\
	Case 2: $y_a = y_b,\; x_a$ is even and $x_b=x_a+1$.\\ 
	In Case 1, for the next adjacent pair ($p_b, p_c$) if $y_b \neq y_c$ then the claim holds because in this case either $c-axis$ or $d-axis$ is updated. On the other hand if $y_b=y_c$ it suffice to show that due to the adjacency of points $p_b$ and $p_c$  $a-axis$ is not updated. Since $x_a$ is odd and $x_b=x_a-1$, the x-coordinate $x_b$ of the point $p_b$ is even. Consequently the x-coordinate $x_c$ of the point $p_c$ is odd $(x_c = x_b-1)$. According to the construction this case will result in the update of $b-axis$. 
	Similarly in Case 2. for the next adjacent pair($p_b, p_c$) if $y_b \neq y_c$ then the claim holds because in this case either $c-axis$ or $d-axis$ is updated. On the other hand if $y_b=y_c$ it suffice to show that due to the adjacency of points $p_b$ and $p_c$ $a-axis$ is not updated. Since $x_a$ is even and $x_b=x_a-1$, the x-coordinate $x_b$ of the point $p_b$ is odd. Consequently the x-coordinate $x_c$ of the point $p_c$ is even $(x_c = x_b+1)$. According to the construction this case will result in the update of $b-axis$.
	
	The selection of axis that is to be updated is done based on the parity of $x$ and $y$ coordinates of the grid points during construction of points in $R^4$. Thus, three consecutive points will never share same axis.
\end{proof}	 

\begin{theorem}
	Let $G$ be a grid in $R^{2}$. $G$ has a Hamiltonian circuit if and only if the set of points $Q$ constructed in $R^{4}$ has a rectilinear tour of $3n$ links. 
\end{theorem}
\begin{proof}
	\textit{Direct implication}: Let $G$ has a Hamiltonian circuit $H$. We construct rectilinear tour in $ R^{4}$.  Consider two consecutive points $p_{i}$ and $p_{j}$ in $H$. By \textit{Lemma \ref{l1}}, the corresponding points $q_{i}$ and $q_{j}$ in $R^{4}$ will share one coordinate and differ all others. Without loss of generality, assume that the points share $a-axis$. Now we connect the points $q_{i}$ and $q_{j}$ using $3$ rectilinear links $L_{1}, L_{2}$ and $L_{3}$ as follows.\\
	$L_{1} = ((a_{i}, b_{i}, c_{i}, d_{i}), (a_{i}, b_{j}, c_{i}, d_{i})),\\ 
	L_{2} = ((a_{i}, b_{j}, c_{i}, d_{i}), (a_{i}, b_{j}, c_{j}, d_{i})),\\
	L_{3} = ((a_{i}, b_{j}, c_{j}, d_{i}), (a_{i}, b_{j}, c_{j}, d_{j})).$\\
	The rectilinear line segment 	$L_{1} $ connects the points $(a_{i}, b_{i}, c_{i}, d_{i})$ and $(a_{i}, b_{j}, c_{i}, d_{i})$ where the second point $(a_{i}, b_{j}, c_{i}, d_{i})$ is not a constructed point.
	We call such a point as an \textit{intermediate point}.
	$L_{1} $ is parallel to $b-axis$. The line segment $L_{2}$ connects two intermediate points $(a_{i}, b_{j}, c_{i}, d_{i})$ and $(a_{i}, b_{j}, c_{j}, d_{i})$ and is parallel to $c-axis$. Finally the line segment  $L_{3}$ connects the intermediate point $(a_{i}, b_{j}, c_{j}, d_{i})$ with the constructed point $(a_{i}, b_{j}, c_{j}, d_{j})$ using a line segment parallel to $d-axis$. Thus a line segment either connects one constructed point with a intermediate point or it connects two intermediate points.  Since the three rectilinear links $L_{1}, L_{2}$ and $L_{3}$ are perpendicular to each other, no two consecutive line segments merge into a single link. Thus any rectilinear line segment covers atmost one constructed point. 
	
	\textit{Reverse implication}: Assume that there exists a rectilinear tour $T$ of $3n$ links in $R^{4}$. Let $q_{i}$ and $q_{j}$  be two consecutive constructed points in $T$.  The intermediate points used to connect $q_{i}$ and $q_{j}$  will share coordinates with  $q_{i}$ and $q_{j}$. For example the intermediate point $(a_{i}, b_{j}, c_{i}, d_{i})$ shares $b-axis$ with $q_{j}$ and all others with $q_{i}$. From \textit{Lemma \ref{l1}} it is clear that any two points will share exactly one coordinte if and only if they are adjacent  in grid and otherwise all four coordinates will be different. Thus the tour is a simple tour and at least $3$ rectilinear links are required to connect $q_{i}$ and $q_{j}$ and a minimum of $3n$ rectilinear links are required to cover $n$ points. Further, if a rectilinear tour of $3n$ links exists in $Q$ then every pair of points $(q_{i}, q_{j})$ in $R^{4}$ must correspond to $(p_{i}, p_{j})$ in $G$ where $p_{i}, p_{j}$ are adjacent. Thus corresponding to the rectilinear tour $(q_{1}, q_{2}... q_{n}, q_{1})$ we have a simple tour $(p_{1}, p_{2}...p_{n}, p_{1})$ where clearly $(p_{1}, p_{2}...p_{n}, p_{1})$ is a Hamiltonian circuit in $G$.	
\end{proof}

\section{Minimum-Link Rectilinear Covering path in $R^{4}$} 
 Minimum Link Rectilinear Covering Tour was shown to be NP-hard in $R^{4}$ in the earlier section. We employ the same construction to construct $Q$, to prove that the Minimum-Link Rectilinear Covering Path is NP-hard in $R^{4}$ where two points $p_{1}$ and $p_{2}$ in  $G$ are given as starting and ending points. The Hamiltonian path in grid graphs problem consists of grid graph defined by a set of grid points and two specified points as the start and end of the Hamiltonian path.
 Hamiltonian path in grid graphs problem was shown to be NP-complete by Itai et al. \cite{itai1982hamilton}. 
 
  Using the same construction method elucidated in Section 3 we construct the set $Q$ of points in $R^{4}$ with two points $q_{1}$ and $q_{2}$ specified as start and end points, which corresponds to the points $p_{1}$ and $p_{2}$ in grid. %Recall that all points are in general position in $Q$. 
   Following theorem shows the reduction.
\begin{theorem}
	Let $G$ be a grid in $R^{2}$. $G$ has a Hamiltonian path if and only if the set of points $Q$ constructed in $R^{4}$ has a rectilinear path of $3(n-1)$ links. 
\end{theorem}
\begin{proof}
	Let $G$ have a Hamiltonian path $P$. We construct rectilinear path in $ R^{4}$. Consider two consecutive points $p_{i}$ and $p_{j}$ in $P$. By Lemma 3.1, the corresponding constructed points $q_{i}$ and $q_{j}$ in $R^{4}$ will share one coordinate and differ in all others. Without loss of generality, assume that the points share $a-axis$. Now we connect the points $q_{i}$ and $q_{j}$ using $3$ rectilinear links $L_{1}, L_{2}$ and $L_{3}$ as follows.\\ 
	$L_{1} = ((a_{i}, b_{i}, c_{i}, d_{i}), (a_{i}, b_{j}, c_{i}, d_{i})),\\ 
	L_{2} = ((a_{i}, b_{j}, c_{i}, d_{i}), (a_{i}, b_{j}, c_{j}, d_{i})),\\
	L_{3} = ((a_{i}, b_{j}, c_{j}, d_{i}), (a_{i}, b_{j}, c_{j}, d_{j})).$\\
	In this case the points $p_{1}$ and $p_{2}$ are the starting and ending points of Hamiltonian path $P$.  So $3(n-1)$ links suffice to create the rectilinear covering path in $R^{4}$.\\
	
	\textit{Reverse implication}: Assume that there exists a rectilinear path of $3(n-1)$ links in $R^{4}$. Consider two consecutive constructed points $q_{i}$ and $q_{j}$ in $Q$. 
	As shown in Lemma 3.1 $q_{i}$ and $q_{j}$ share at most one axis and hence $3$ rectilinear links are required to connect $q_{i}$ and $q_{j}$. Thus, a minimum of $3(n-1)$ rectilinear links are required to cover $n-1$ edges.
	Further, if a rectilinear path  $P^*$ of $3(n-1)$ links exists in $Q$ then every $(q_{i}, q_{i+1})$ in $P^*$ in $R^{4}$ must correspond to $(p_{i}, p_{i+1})$ in $G$ where $p_{i}, p_{i+1}$ are consecutive points in $P$. 
	Thus, corresponding to the rectilinear path $(q_{1}, q_{i1}... q_{i~n-2}, q_{2})$ we have a simple path $(p_{1}, p_{i1}...p_{i~n-2}, p_{2})$ where clearly $(p_{1}, p_{i1}...p_{i~n-2}, p_{2})$ is a Hamiltonian path in $G$.	
\end{proof}

\section{Conclusions}

We prove that Minimum-Link Rectilinear Covering Tour is NP-hard in $R^{4}$ when the points are in relaxed general position. We also prove that Minimum-Link Rectilinear Covering Path for a given pair of start and end points is NP-hard in $R^{4}$ when the points are in relaxed general position.

\section{Acknowledgements}
Indulekha TS participated in the discussions.

\end{document}